\documentclass[format=acmsmall, review=false]{acmart}
\usepackage{acm-ec-25}
\usepackage{booktabs} 
\usepackage{tcolorbox}
\usepackage{hyperref}

\usepackage{latexsym,amsmath,amsthm,graphicx,enumerate}
\usepackage[ruled]{algorithm2e} 

\SetAlFnt{\small}
\SetAlCapFnt{\small}
\SetAlCapNameFnt{\small}
\SetAlCapHSkip{0pt}
\IncMargin{-\parindent}

\newcommand{\shm}{\textsc{SHM}}
\newcommand{\cacq}{\textsc{CA-CQ}}
\newcommand{\smf}{\textsc{SMF}}

\newcommand{\HH}{\mathcal{H}}
\newcommand{\EE}{\mathcal{E}}
\newcommand{\E}{\mathcal{A}}
\renewcommand{\SS}{\mathcal{S}}
\newcommand{\CC}{\mathcal{C}}

\newtcolorbox{searchproblembox}[1]{
  colframe=black, colback=white, coltitle=white,
  colbacktitle=black, fonttitle=\bfseries,
  boxrule=1pt, arc=3mm, width=0.9\textwidth,
  left=5pt, right=5pt, top=5pt, bottom=5pt,
  title={#1}
}
\newtheorem{claim}{Claim}

\setcitestyle{authoryear}
\newcommand{\PP}{\mathcal{P}}

\title[Near-Feasible Solutions to Complex Stable Matching Problems]{Near-Feasible Solutions to Complex Stable Matching Problems}

\author{Gergely Csáji (HUN-REN KRTK)}

\begin{abstract}
 In this paper, we demonstrate that in many NP-complete variants of the stable matching problem, such as the Stable Hypergraph Matching problem, the College Admission problem with Common Quotas and the Stable Multicommodity Flow problem, a near-feasible stable solution -- that is, a solution which is stable, but may slightly violate some capacities -- always exists. Our results provide strong theoretical guarantees that even under complex constraints, stability can be restored with minimal capacity modifications. 

To achieve this, we present an iterative rounding algorithm that starts from a stable fractional solution and systematically adjusts capacities to ensure the existence of an integral stable solution. This approach leverages Scarf's algorithm to compute an initial fractional stable solution, which serves as the foundation for our rounding process. Notably, in the case of the Stable Fixtures problem, where a stable fractional matching can be computed efficiently, our method runs in polynomial time. 

These findings have significant practical implications for market design, college admissions, and other real-world allocation problems, where small adjustments to institutional constraints can guarantee stable and implementable outcomes.
\end{abstract}

\begin{document}

\begin{titlepage}

\maketitle

\vspace{1cm}
\setcounter{tocdepth}{2} 
\tableofcontents

\end{titlepage}

\section{Introduction and Related Work}
Stable matching problems are fundamental in economics and operations research, underpinning key applications such as student admissions, job markets, kidney exchanges, and housing allocations. These models provide a structured way to assign agents to institutions while ensuring fairness and stability.

The Stable Matching problem was introduced by \citet{gale1962college}, along with their celebrated polynomial-time algorithm to find a stable matching.  In their model, students express preferences over colleges, while colleges rank applicants and specify admission quotas, all of which are coordinated centrally. The key solution concept they introduced is \emph{stability}, which ensures that a college only rejects an applicant if it has already filled its available spots with more preferred candidates. They proposed the deferred-acceptance (DA) algorithm, which guarantees the existence of a stable matching and yields an outcome most favorable to students when they are the proposing side. 

Gale and Shapley’s work has sparked extensive research across multiple disciplines, including mathematics, computer science, operations research, economics, and game theory. This field, commonly referred to as Matching Under Preferences or Matching Markets, has since grown into a major area of study. For a comprehensive overview, we refer readers to the book by \citet{manlove2013algorithmics}.
Stable matchings are also utilized in various real-world applications, including the famous U.S. National Resident Matching Program (NRMP) \cite{roth1984}. Furthermore, they are used in various national college admissions and school choice programs worldwide. For an overview of applications in the U.S., we refer readers to the survey by \citet{roth2008}.

Since the introduction of stable matchings, extensive generalizations have been studied, encompassing more complex settings where additional constraints, such as quotas or preferences over multiple partners, come into play. Some variants remain tractable, such as the Stable Roommates problem \cite{irving1985efficient}, which extends stable matchings to general graphs and the Stable Fixtures problem \cite{Irving2007fixtures}, which also introduces capacities. However, many extensions lead to computational intractability, including Stable Hypergraph Matching, College Admission (or Many-to-one Stable Matching) with Common Quotas, and the Stable Multicommodity Flow problem, all of which are NP-hard. Also, for such complex variants, stable matchings may fail to exist, making them less viable in practice. The main purpose of this paper is to show that even in these NP-hard variants, the existence of stable matchings can be guaranteed, if the capacities are allowed to be slightly violated.

As a particularly relevant case, we also focus on a problem coming from college admissions. Here, many additional constraints and special features necessitate modifications to the standard DA algorithm. One key challenge is the presence of common quotas, which arise for example due to faculty-imposed limits on program enrollments and national quotas for state-financed students in specific subjects, see a description of the Hungarian case in \cite{agoston2022college}. This feature makes the underlying matching problem NP-hard \cite{biro2010college} (except for the case, when the common quotas form a laminar system), though the resulting computational challenges can be potentially tackled by Integer Programming methods, as demonstrated by \citet{agoston2022college}. As one of our main results, we show that for the College Admission with Common Quotas problem, if each college belongs to only a few sets with a common quota, then a minimal amount of modifications in the quotas can guarantee the existence of a stable matching.


The study of stable matchings in hypergraphs originated with  \citet{AharoniFleiner03}, who formulated it as the problem of identifying core elements in simple NTU-games. This is closely related to core stability in hedonic coalition formation games (see, e.g., \cite{banerjee2001core, woeginger2013core, aziz2012existence}). In this framework, vertices represent agents, while hyperedges correspond to potential coalitions. Each agent can participate in at most one (or more generally, a fixed number of) coalition(s) in the final solution, meaning that the goal is to find a valid matching in the hypergraph. Since agents' preferences depend only on their own coalition, the problem is referred to as hedonic. Preferences are typically represented by a linear ordering of the incident hyperedges for each vertex. A coalition structure belongs to the core (or equivalently, the matching is stable) if no hyperedge exists that is strictly preferred by all its members over their current assignments.

A special case of stable hypergraph matching, the 3-dimensional variant of stable matching -- that is also related to the problem of assigning residents to hospitals in the presence of couples, as is the case currently in the NRMP -- was proposed by 
\citet{knuth1997mariages} and later investigated by 
\citet{doi:10.1137/0404023} and 
\citet{huang2007two}. They proved that the problem is NP-complete even in restricted settings.

Subsequent research has explored more complex matching problems to determine whether they can be framed as special cases of stable matchings in hypergraphs, where agents may also have capacity constraints. A general framework for this approach was developed by \citet{biro2016fractional}, who applied Scarf’s lemma \cite{scarf1967core}, a fundamental theorem for NTU games, to prove the existence of fractional stable solutions in (capacitated) hypergraph matching problems. Moreover, as Scarf's lemma was proved in an algorithmic way -- later referred to as Scarf's algorithm -- this implies that stable fractional matchings can be found algorithmically.

Sadly, the computation of a fractional solution by Scarf'S algorithm is not polynomial time in general. The PPAD-hardness of computing a fractional stable matching was proved by \citet{kintali2013reducibility} and were linked with several other PPAD-hard combinatorial problems.

A different type of generalization to network flows, called the Stable Flow problem, introduced by \citet{inproceedings}, captures the challenge of finding a flow that satisfies stability conditions within a network, and the author provided a polynomial-time algorithm for computing such a stable solution. This concept was later extended by \citet{kiraly}, who explored the more complex case of multiple commodities. In this framework, they introduced orderings on the network arcs to rank the commodities that can be sent along each arc. They showed that a stable fractional multicommodity flow always exists; however, finding such a solution is PPAD-hard, highlighting the computational challenges involved. The problem of determining whether a stable multicommodity flow exists was proven NP-complete by \citet{cseh}, and further work by \citet{csaji2022complexity} demonstrated that even with just two commodities, the problem remains NP-complete. However, the work of \citet{csaji2022complexity} also proved that Scarf's algorithm can be used to find a stable fractional solution.

It has been observed previously that rounding the fractional solution obtained from Scarf's algorithm can lead to near-feasible stable solutions, meaning a stable matching can be achieved if capacities can be slightly adjusted. This approach has been explored in various settings. For instance, \citet{nguyen2018near} studied the Hospital-Resident matching problem with couples—an NP-hard extension of the many-to-one stable matching problem—where some agents have joint preferences over pairs of positions. They demonstrated that modifying hospital capacities by at most two ensures the existence of a stable outcome via an iterative rounding algorithm. In a restricted case, where couples can only apply together for the same positions, \citet{dean} showed that a modification of one in the capacities is already enough, and furthermore a near-feasible stable matching can be found in polynomial-time. Most recently, \citet{csaji2023couples} proved that when the couples' preferences satisfy some additional restrictions, but still allowing them to apply for different positions, a near-feasible stable matching can be found in polynomial time and a change of one in the capacities is enough for that.
Apart from resident allocation with couples, \citet{nguyen2019stable} showed that allowing small capacity perturbations can help find stable matchings that satisfy proportionality constraints. \citet{nguyen2021stability} used a similar approach to find group-stable matchings. In a related context, \citet{chencsaji} examined capacity adjustments in school choice, aiming to achieve matchings that are both stable and either Pareto-efficient or perfect (where every student is assigned a seat). More generally, capacity modifications have been proposed as a practical tool to overcome non-existence and inherent computational challenges in complex matching markets, while maintaining desirable properties.


In this paper, we build upon and extend these results, demonstrating that small capacity modifications can ensure the existence of stable solutions for a new range of problems. Our results highlight the practical significance of \emph{near-feasible stable matchings}—offering guarantees that even under complex quota constraints, stability can be restored with only minimal adjustments. This insight has profound implications for policy design in college admissions, workforce allocation, and market design, where slight modifications to institutional constraints can lead to stable, implementable outcomes.

\subsection{Contributions}

Our contributions establish existential and algorithmic guarantees for near-feasible stable solutions in a range of complex stable matching problems, extending the techniques introduced by \cite{nguyen2018near,nguyen2019stable}.
\begin{itemize}
\item
\emph{Stable Hypergraph Matching (with Vertex Capacities):} We show that if the largest hyperedge in the hypergraph has size $\ell$, then modifying each capacity by at most $\ell - 1$ ensures the existence of a stable matching. This result holds even when ties are present. Furthermore, the changes in the total sum of the capacities remain bounded by $\ell - 1$. For the Stable Fixtures problem, where $\ell = 2$, we provide a polynomial-time algorithm to compute such a modified instance and a corresponding near-feasible stable matching. 

\item 
\emph{College Admission with Common Quotas:} We prove that when each college belongs to at most $\ell$ subsets with common quotas, modifying each individual and common quota by at most $2\ell - 1$ ensures the existence of a stable matching, even in the presence of ties. For example in Hungary, common quotas arise from the faculty quotas imposed on their programs and from the national quotas set for state-financed students in each subject (see \cite{agoston2022college}), which implies $\ell = 2$, and hence the maximum modification required is at most 3, which makes our result very practical and desirable in complex College Admission mechanisms, where stability is essential. 

\item \emph{Stable Multicommodity Flow:} We extend our analysis to network flow settings, showing that when each commodity-specific capacity remains fixed, adjusting aggregate capacities by at most $k - 1$ (where $k$ is the number of commodities) guarantees a stable solution. Furthermore, the size of the flows can only change by strictly less than 1 and even the change in the size of the aggregate flow can be bounded by 1, if we allow a change of 2 in the individual flows. This result provides insights into designing robust, nearly stable network flow mechanisms in practical applications such as traffic and logistics systems.

\end{itemize}

As in most practical cases, the parameters $\ell$ and $k$ are usually quite small, these results reinforce the practical viability of stable matching models in constrained settings, demonstrating that small, targeted modifications to system parameters can ensure desirable outcomes. Our findings are particularly relevant for economists and policymakers designing market mechanisms, offering constructive methods to restore stability without major systemic changes.

\section{Preliminaries}

In this section, we formally define our models and introduce the necessary notations.

We use $\mathbb{Z}_+$ and $\mathbb{R}_+$ to denote the sets of nonnegative integers and nonnegative real numbers, respectively. For an integer $ k \in \mathbb{Z}_+ $, we define $[k]$ as the set $\{ 1,2,\dots, k\}$.

\medskip
\textbf{Hypergraphs and Matchings.}
A \emph{hypergraph} $\mathcal{H} = (V, \mathcal{E})$ consists of a set of vertices $V$ and a set of (hyper)edges $\mathcal{E} \subseteq 2^V \setminus \{ \emptyset \}$, where $2^V$ denotes the power set of $V$. A \emph{graph} is a special case of a hypergraph in which each hyperedge contains exactly two vertices. For the sake of simplicity, we will refer to hyperedges simply as edges too.

In economic markets, such as school-choice, resident-allocation, or team-formation, the set of vertices corresponds to the set of agents, while the set of edges describes the acceptability relations or possible contracts.

Given a set of capacities $q:V\to \mathbb{Z}_+$, a function $M: \mathcal{E} \to [0,1]$ is called a \emph{fractional matching} if it satisfies
\[
\sum_{e \ni v} M(e) \leq q(v) \quad \text{for all } v \in V.
\]
Similarly, a \emph{matching} is a function $M: \mathcal{E} \to \{0,1\}$, satisfying the same capacity constaints. 

One can also think of a matching, as a subset of edges such that no vertex $v \in V$ appears in more than $q(v)$ edges of $M$.
A vertex $v$ is called \emph{unsaturated} in $M$ if the inequality is strict, i.e.,
$\sum\limits_{e \ni v} M(e) < q(v),$
and \emph{saturated} if equality holds.
Given a (fractional) matching $M$, we let $M(v): =\{e\in \EE\mid M(e)>0\}$ and $|M(v)|: = \sum\limits_{e\ni v}M(e)$.

The \emph{characteristic vector} of a (fractional) matching $M$ is the vector $x\in \mathbb{R}^{\EE}$, where $x[e]=M(e)$.

\medskip
\textbf{Directed Graphs and Flows.}
A \emph{directed graph} $D = (V, \E)$ consists of a set of vertices $V$ and a set of directed arcs $\E \subseteq V \times V$. Each arc $uv \in \E$ starts at vertex $u$ and ends at vertex $v$. For a vertex $v$, we denote the set of incoming arcs by $\rho(v)$ and the set of outgoing arcs by $\delta(v)$.

In a directed graph $D = (V, \E)$ with two distinguished vertices $s, t \in V$ and a capacity function $c: \E \to \mathbb{Z}_+$, a function $f: \E \to \mathbb{R}_+$ is called a \emph{flow} if it satisfies:
\begin{enumerate}
    \item \emph{Kirchoff law:} For every $ v \in V \setminus \{s, t\} $, 
    \[
    \sum_{a \in \delta(v)} f(a) = \sum_{a \in \rho(v)} f(a).
    \]
    \item \emph{Capacity constraint:} For every $a \in \E$,
    \[
    f(a) \leq c(a).
    \]
\end{enumerate}

The \emph{size of a flow $f$} is given by $$|f^j|:=\sum\limits_{a\in \delta (s)}f^j(a).$$

We refer to the tuple $(D, s, t, c)$ as a \emph{directed network}.

In the \emph{multicommodity flow} framework, we consider multiple commodities, indexed by $ j \in [k] $. Each commodity has its own source-terminal pair $(s^j, t^j)$ and a commodity-specific capacity function $ c^j: \E \to \mathbb{Z}_+ $. A collection of functions $ f = (f^1, \dots, f^k) $ is called a \emph{multicommodity flow} if:
\begin{enumerate}
    \item Each $ f^j $ is a valid flow in the network $ (D, c^j, s^j, t^j) $, and  
    \item The total flow on any arc does not exceed its overall capacity, i.e.,  
    \[
   f(a):= \sum_{j \in [k]} f^j(a) \leq c(a) \quad \text{for all } a \in \E.
    \]
\end{enumerate}

For a multicommodity flow $f=(f^1,\dots, f^k)$, the \emph{size of the aggregate flow $f$} is given by $$|f|:=\sum\limits_{j\in [k]}|f^j|.$$

Multicommodity flows model scenarios where multiple independent flows compete for limited resources, such as traffic routing in communication networks, logistics and supply chain distribution, or scheduling in transportation systems. 

\medskip
\textbf{Polyhedra.}
A 	\emph{polyhedron} is a set of points in $\mathbb{R}^n$ defined by a finite number of linear inequalities of the form $Qx \leq d$, where $Q$ is an $m \times n$ matrix and $d$ is an $m$-dimensional vector. A 	\emph{polytope} is a bounded polyhedron.

A row of a constraint matrix $Q$ is called \emph{tight} for a solution $x$ if the corresponding inequality holds with equality, i.e., $Q_i x = d_i$. A solution satisfying a set of constraints tightly is said to be on the corresponding 	\emph{face} of the polyhedron.

An 	\emph{extreme point} $z$ of a polyhedron $\mathcal{P}$ is a point that cannot be expressed as a strictly convex combination of two other distinct points in the polyhedron, i.e., there are no $z_1,z_2\in \mathcal{P}$ and $\lambda\in (0,1)$ such that $z=\lambda z_1 + (1-\lambda )z_2$. In other words, an extreme point is a vertex of the polyhedron.

We proceed to define our central problems.

\subsection{Stable Hypergraph Matching}

First, we define the Stable Hypergraph Matching (\shm) problem with ties and capacities. Let $\mathcal{H}=(V,\EE)$ be a hypergraph, where each vertex $v\in V$ has a capacity $q(v)\in \mathbb{Z}_+$ and a weak preference order $\succcurlyeq_v$ over the edges containing $v$. For a weak preference order $\succcurlyeq_v$, we use the notation $e\succ_vf$ to denote that $e$ is strictly preferred to $f$ by $v$, $e\sim_v f$ to denote that $e$ and $f$ are tied (or $e=f$) and $e\succcurlyeq_v f$ to denote that either $e\succ_v f$ or $e\sim_v f$.

\begin{definition}
A (fractional) matching $M$ is \emph{ blocked} by an edge $f$, if $M(f) < 1$ and, for every $v \in f$, either $v$ is unsaturated by $M$ or there exists an edge $f_v \in M(v)$ such that $f_v \prec_v f$.
\end{definition}

Intuitively, blocking means that the edge (e.g., a contract) is strictly better for each participating agent in it, who is already at full capacity. Hence, every vertex (agent) in the edge has an incentive to deviate in this case.

\begin{definition}
A (fractional) matching $M$ is \emph{stable} if no edge blocks it.
\end{definition}

An instance of {\sc Stable Hypergraph Matching} (\shm) is given by the tuple $(\HH, (q(v), \succcurlyeq_v)_{v\in V})$, where the goal is to find a stable matching $M$. The {\sc stable fixtures} problem is the restriction of \shm\ to graphs.

\shm\ models scenarios where no group of agents (represented as hyperedges) has an incentive to deviate, ensuring stability in multi-agent allocation problems such as group project assignments and team formation. \shm\ can also incorporate several stable matching problems, including dual admission problems, where students apply to both university courses and internships at companies, as well as the restricted case of the Hospital-Resident allocation problem with couples, given that no couple submits applications to the same position.

\subsection{College Admission with Common Quotas}

The College Admission problem with Common Quotas (\cacq) extends the classical college admission model by introducing quota constraints that apply across multiple colleges. 

Here, we are given a set of colleges $\CC = \{c_1, \dots, c_m\}$ and a set of students $\SS = \{s_1, \dots, s_n\}$. Each college $c_i$ has a quota $q(c_i)$ and a weak preference order $\succcurlyeq_{c_i}$ over acceptable students. Each student $s_k$ has a weak preference order $\succcurlyeq_{s_k}$ over acceptable colleges. These acceptability relations are modeled by a bipartite graph $G=(V; \EE)$ with vertex set $V=\SS\cup \CC$.

Additionally, we are given sets of colleges $C_1, \dots, C_N$, each with a common quota of $q(C_j)$. Each $C_j$ is a subset of $\mathcal{C}$, and a college may belong to multiple such sets. Each set $C_j$ follows a weak \emph{master preference list} $\succcurlyeq_{C_j}$, which is consistent with individual college preferences, that is, $s_k \succ_{c_i} s_{k'}$ for some $c_i \in C_j$ implies $s_k \succ_{C_j} s_{k'}$ and $s_k \sim_{c_i} s_{k'}$ for some $c_i \in C_j$ implies $s_k \sim_{C_j} s_{k'}$.

A \emph{(fractional) matching} $M$ is defined here by the inequalities
\begin{itemize}
    \item[--] $|M(s_k)|\le 1$ for $s_k\in \SS$,
    \item[--] $|M(c_i)|\le q(c_i)$ for $c_i\in \CC$ and
    \item[--] $\sum\limits_{c_l\in C_j}|M(c_l)| \leq q(C_j)$  for every $C_j\in \{ C_1,\dots ,C_N\}$.
\end{itemize}
That is, a matching satisfies all individual and common quotas and assigns each student to at most one college.

It will be more convenient to assume that for each college $c_i$, there exists a set $C_j=\{ c_i\}$ with common quota $q(C_j)=q(c_i)$ and $\succcurlyeq_{C_j}=\ \succcurlyeq_{c_i}$. Hence, from now on, we only focus on satisfying the common quota constraints.

Also for convenience, it is easier to think of the preferences $\succcurlyeq_{s_k}$ over the colleges as preferences over the incident edges to $s_k$ in $G$. Since $s_k$ and $c_i$ uniquely determine the edge $\{ s_k,c_i\}\in \EE$, this is well-defined.

\begin{definition}
An edge $e=\{ s_k, c_i\}\in \EE$ \emph{blocks} a (fractional) matching $M$ if:
\begin{itemize}
    \item[--] $|M(s_k)|<1$ or $e\succ_{s_k}e'$ for some $e'=\{ s_k,c_{i'}\} \in M(s_k)$ (for integral matchings, $M(s_k)$ has at most one element), 
    
    \item[--] any set $C_j$ containing $c_i$ (including $\{ c_i\}$) is either below its quota (i.e., $\sum\limits_{c_l\in C_j}|M(c_l)|<q(C_j)$) or has a student $s_{k'}\in \bigcup\limits_{c_l\in C_j}M(c_l)$ assigned to some $c_{l} \in C_j$ such that $s_k \succ_{C_j} s_{k'}$.
\end{itemize}
\end{definition}

Intuitively, blocking means that both the student strictly prefers to college to his current allocation, and for each set $C_j$ the college is included in, either the student can argue that the corresponding common quota is not reached, so there are empty seats left, or otherwise that a student worse than him is admitted there. 

\begin{definition}
A (fractional) matching $M$ is \emph{stable} if no edge $\{ s_k,c_i\}$ blocks $M$.
\end{definition}

An instance of the {\sc College Admission with Common Quotas} (\cacq) consists of the tuple $(G=(\SS\cup \CC,\EE),(C_i,q(C_i)_{i\in N}),(\succcurlyeq_v)_{v\in \SS\cup \CC})$ and the goal is to find a stable matching.

Stability in college admissions ensures that no student and college prefer each other over their current assignment, while maintaining global capacity constraints across multiple institutions.

Common quotas are often imposed in national university admission mechanisms, for example by restricting the total number of students that can be admitted to each department, each university and also across multiple institutions, such as limiting the total number of students admitted to a specific major nationwide. Furthermore, there are usually different quotas for self-financed and state-financed seats.

\subsection{Stable Multicommodity Flow}

The Stable Multicommodity Flow (\smf) problem models competition among multiple commodities for network resources. Given a directed graph $D = (V, \E)$ with sources $s^1, \dots, s^k$ and sinks $t^1, \dots, t^k$ (one per commodity), each vertex $v \in V$ has a weak order $\succcurlyeq_v^j$ on its incident arcs (i.e., $\rho (v)\cup \delta (v)$) for each commodity $j$. Each arc $a \in \E$ has a weak order $\succcurlyeq_a$ over commodities and integral capacities $c(a)$ and $c^j(a)$, where $c(a)$ bounds total flow and $c^j(a)$ bounds flow of commodity $j$. Intuitively, the ranking $\succcurlyeq_v^j$ represent an agent's preference over which routes or companies he prefers to send or recieve the given commodity, and the ranking $\succcurlyeq_a$ represents the preferences of the transporting agency over the commodities it is willing to deliver. For example,  a transport agency may prefer to deliver non-hazardous and non-perishable goods over those that pose logistical challenges.

In such a model, the blocking structures that may arise, are walks, along which the participating agents would have an incentive to reroute some flow. Here, a walk $W=(v_1,a_1,$  $v_2,\dots, $ $a_{l-1},v_l)$ denotes an alternating set of vertices and arcs, such that for all $i\in [l-1]$, $a_i=v_iv_{i+1}$. 

\begin{definition}
A walk $W=(v_1,a_1,v_2,\dots, a_{l-1},v_l)$ \emph{blocks} a flow $f=(f^1,\dots, f^k)$ for commodity $j$ if:
\begin{enumerate}
    \item $f^j<c^j(a_i)$ for $i\in [l-1]$, that is, every arc has available capacity for commodity $j$,
    \item $v_1= s^j$ or there exists some arc $b\in \delta (v_1)$ such that $f^j(b)>0$ and $a_1\succ_{v_1}^jb$, so $v_1$ can either send more flow of commodity $j$, or reroute some outgoing flow from a worse arc,
    \item $v_l=t^j$ or there exists some arc $b\in \rho (v_l)$ such that $f^j(b)>0$ and $a_{l-1}\succ_{v_l}^jb$, so $v_l$ can either send more flow of commodity $j$, or reroute some incoming flow from a worse arc; 
    \item if $\sum_{j\in [k]}f^j(a_i)=c(a_i)$, then there exists $j'\ne j$ such that $f^{j'}(a_i)>0$ and $j\succ_{a_i}j'$, that is, if an arc is saturated, then another commodity $j'$ on $a_i$ is less preferred than $j$.
\end{enumerate}
\end{definition}

\begin{definition}
A multicommodity flow $f$ is \emph{stable} if no blocking walk exists for any commodity.
\end{definition}

An instance of \smf\ consists of a tuple $(D,(s^i,t^i)_{i\in [k]},c,(c^i)_{i\in [k]}, (\succcurlyeq_v^j)_{j\in [k],v\in V},(\succcurlyeq_a)_{a\in \E})$ and the goal is to find a stable multicommodity flow.

The study of stable multicommodity flows is crucial for various real-world applications where resources need to be distributed efficiently and fairly, such as in transportation networks, communication systems, and supply chains. Understanding the computational complexity of this problem and identifying feasible methods to compute stable solutions is of significant practical importance, especially in scenarios where multiple types of goods or services must coexist within the same infrastructure.
\subsection{Scarf's Lemma}
Finally, we state Scarf's key Lemma.

\begin{lemma}[\cite{scarf1967core}] Let $Q$ be an $n\times m$ nonnegative matrix, such that every column of $Q$ has a nonzero element and let $d\in \mathbf{R^n_+}$. Suppose that every row $i$ has a strict ordering $\succ_i$ on those columns $j$ for which $Q_{ij}>0$. Then there is an extreme point of $\{ Qx\le d, \; x\ge 0\}$, that dominates every column in some row, where we say that $x\ge 0$ dominates column $j$ in row $i$, if $Q_i x=d[i]$ and $k\succcurlyeq_i j$ for all $k\in \{ 1,\dots  ,m\}$, such that $Q_{ik} x_k>0$. Also, this extreme point can be found algorithmically.
\end{lemma} 

Scarf's Lemma gives existential guarantees for fractional stable solutions in a large set of stable matching problems.

\section{Near-feasible Stable Hypergraph Matchings}


First, we consider the Stable Hypergraph Matching problem, \shm. Our purpose is to find new capacities $q'$ close to the original ones together with a corresponding stable matching $M$ for the new capacities.

\subsection{Preparations}
To begin with, suppose that our instance of \shm\ includes ties in some preference lists. Then, in the beginning we break the ties in an arbitrary way. On one hand, the stable (fractional) matchings of this new instance are stable for the original instance too, because if an edge blocks a matching $M$ in the original instance, then it still blocks $M$ after the tie breaking. This holds because if $e\succ_v f$ then $e$ and $f$ are not tied, so breaking the ties maintain that $e\succ_v f$. 
On the other hand, if we find new $q'$ capacities, such that there is a stable matching $M$ with respect to the capacities $q'$ after the tie-breaking, then the same matching is also stable with respect to the capacities $q'$ in the original instance, by the same reasoning.

Therefore, we obtain that it is enough to work with strict preferences $\succ_v$ for our purposes. Hence, in the remainder of this section, we suppose that all preference lists are strict.

We continue by describing how to obtain a stable fractional matching.
To do this, for an instance $I=(\HH,  (q(v),\succ_v)_{v\in V})$ of \shm\ (with strict preferences), we create a polyhedron satisfying the conditions of Scarf's lemma the following way.

First, we create a matrix $Q$ that consists of the incidence matrix $A$ of the hypergraph $\mathcal{H}$ with an additional identity matrix at the bottom. Each row in $A$ corresponds to a vertex $v\in V$ and each row in the bottom identity matrix corresponds to an edge $e\in \EE$. Let us denote these rows by $A_v$ and $I_e$ respectively.

Let the bounding vector $d$ be defined such that it is the capacity $q(v)$ in vertex $v$'s row, and $1$ in the rows corresponding to the bottom identity matrix.

For the rows of the bottom identity matrix, there is only a single nonzero element, so their preference over the nonzero entries is trivial.
The strict preferences of the rows corresponding to vertices $v\in V$ are created according to $\succ_v$ (which we assumed to be strict). Note that since the nonzero elements correspond exactly to the edges containing $v$, this is well-defined. 

It is straightforward to verify, that the characteristic vectors of the (fractional) stable matchings of $\mathcal{I}$ are in a one-to-one correspondence with the dominating (not necessarily extreme) points of the polyhedron given by $ \{  Qx\le d,x\ge 0\} = \{  Ax\le q, 0\le x\le 1\}$.
Recall that we say that $x\ge 0$ dominates column $j$ in row $i$, if $Q_i x=d[i]$ and $k\succcurlyeq_i j$ for all $k\in \{ 1,\dots  ,m\}$, such that $Q_{ik} x_k>0$ and a point is dominating, if it dominates every column in at least one row.

Hence, we can find a stable fractional solution with Scarf's algorithm.

Next, we argue that we can assume that all fractional stable matchings saturate every vertex $v\in V$. Hence, we can impose $Ax=q$ (that is, add $-Ax\le -q$ to the inequalities) in the beginning, and thus our polyhedron will be $\mathcal{P}=\{ Qx\le d, x\ge 0\} =\{ Ax=q,0\le x\le 1\}$.

\begin{claim}
\label{claim:hyp-saturated}
We can suppose that all fractional stable matchings in the starting instance $I$ satisfy $|M(v)|=q(v)$ for $v\in V$.
\end{claim}
\begin{proof}

To prove the claim, for each vertex $v\in V$, we add edges $e_v^1,\dots, e_v^{q(v)}$ that are strictly worst for $v$ in this order. All these edges contain the single vertex $v$. Let this new hypergraph be $\HH'=(V,\EE')$.

It is clear that in this new instance, any fractional stable matching saturates all vertices. If some $v\in V$ is left unsaturated, then there exists some edge $e_v^j=\{ v\}$ such that $M(e_v^j)<1$, so $e_v^j$ blocks $M$, a contradiction.

What we have to show for our purposes is that if we find capacities $q'$ and a stable matching $M'$ with respect to the capacities $q'$, then deleting all edges of the form $e_v^j$ leads to a matching $M$ that is stable in the original hypergraph $\HH =(V,\EE)$ too with respect to the capacities $q'$.

Suppose that this matching $M$ admits a blocking edge $e\in \EE$. As $e$ did not block $M'$ in $\EE'$ with respect to $q'$, there was a vertex $v\in V$, such that $|M'(v)| =q'(v) $ and for each $f\in M'(v)$, we had $f\succ_v e$. If $|M(v)|=q'(v)$ holds too, then we get that $e$ does not block by the previous observation. Else, there was an edge $e_v^j$ with $M'(e_v^j)>0$. However, by the construction, $e\succ_v e^j_v$, contradicting $f\succ_v e$ for all $f\in M'(v)$.
\end{proof}

\subsection{The Iterative Rounding algorithm}

We start with a structural observation. Take the constructed polyhedron $\{ Qx\le d, x\ge 0\}$.

\begin{lemma}
\label{lemma:hyp-1}
Let $x^*$ be the dominating solution found by running Scarf's algorithm on $\{  Qx\le d, x\ge 0\}$. Let $y$ be an integer vector such that, if $x^*[e]=0$ or $1$, then $y[e]=0$ or $1$, respectively.
Furthermore, let 
$q'(v)=\begin{cases} A_v y & \text{ if } A_v x^*=q(v)\\  \max\{ q(v), A_v y\} & \text{ 
otherwise} 
\end{cases}$. Then the matching $M$ whose characteristic vector is $y$ is stable and feasible with respect to $q'$.
\end{lemma}
\begin{proof}
By the definition of the capacities $q'$, it is immediate that $M$ respects all capacity constraints and thus is a matching.

Let us suppose for the contrary that $M$ is not stable. Then, there exists an edge $e\in \mathcal{E}$ that blocks $M$. 

Since $x^*$ was a dominating solution, it must have dominated column $e$ in some row. If it dominated $e$ in row $I_e$, then $x^*[e]=1=y[e]=M(e)$, contradicting that $e$ blocks. Otherwise, it was dominated in some row $A_v$ such that $v\in e$. This implies that $A_v x^* = q(v)$, so $A_v y = q'(v)$, and hence $v$ is saturated in $M$. Furthermore, for every $f_v\in \mathcal{E}$ such that $v\in f_v$ and $x^*[f_v]>0$, $f_v$ was weakly better for $v$ than $e$ (which implies either $f_v\succ_ve$ or $f_v=e$). However, then by the definition of $q'$, and that  $y[f]=0$, whenever $x^*[f]=0$, we get that $v$ is saturated in $M$ and weakly prefers every $f_v\in M(v)$ to $e$, we get that $e$ cannot block $M$, contradiction again. 
\end{proof} 

Let $\ell \ge 2$ be the size of the largest edge in $\mathcal{H}$. That is, $|e|\le \ell$ for all $e\in \mathcal{E}$.

Our iterative rounding algorithm is described in Algorithm~\ref{alg:shm}. First, it computes a dominating solution $x^*$ by Scarf's algorithm and sets $z=x^*$.
To bound the change in the aggregate capacity, we add an aggregate capacity row to the incidence matrix $A$, that is a row $\sum_e x[e] \cdot |e| = \sum_{v\in V} q(v)$. By Claim~\ref{claim:hyp-saturated}, we can assume that every row in $x^*$ is tight, so $x^*$ satisfies this as $\sum_{e\in \EE}x^*[e]\cdot |e|=\sum_{v\in V}A_vx^* =\sum_{v\in V} q(v)$ .

Then, until our vector is not integral, we check whether there is a row we can safely remove without worrying about the capacities being violated too much in the future. For this, we delete a row $A_v$, if $A_v(\lceil z \rceil - \lfloor z \rfloor ) \le \ell$ and if there is no such row, but there is only one remaining fractional component, then we delete the aggregate capacity row.

We will show that this algorithm is guaranteed to terminate.

\begin{algorithm}[t]
	\SetAlgoNoLine
	\KwIn{An instance of \shm.}
	\KwOut{A stable matching $M$ with respect to some capacities $q'$ satisfying $|q'(v)-q(v)|\le \ell -1$ and $|\sum\limits_{v\in V}(q'(v)-q(v))|\le \ell -1$.}
	  Run Scarf's algorithm on $\PP =  \{ Qx\le d,x\ge 0\}$ and obtain a dominating solution  $x^*$. 

    Initialize $z = x^*$.

    Add an aggregate capacity row to the incidence matrix $A$, that is a row $\sum\limits_{e\in \EE} x[e] \cdot |e| = \widetilde{q}(V) := \sum\limits_{v\in V} q(v)$.
    
	\While{$z$ is not integral}{
		\If{there is a row $v$ such that $A_v(\lceil z \rceil - \lfloor z \rfloor ) \le \ell$}{
			Delete it.
		}
		\ElseIf{there is no such row but there are at most one $z[e]$ component that is fractional}{
			Delete the aggregate capacity constraint. 
		}
		Consider the polyhedron given by the remaining constraints of $\PP$ along with the conditions $x[e] = z[e]$ if $z[e]$ is integer, and find the extreme point $z^*$ that maximizes $\sum_{e\in \EE}z[e]\cdot |e|$.

    	Update $z := z^*$.
    }
    Output the obtained matching $M$ corresponding to $z$ and the $q'$ capacities as defined in Lemma \ref{lemma:hyp-1}.
	\caption{Iterative Rounding for \shm}
	\label{alg:shm}
\end{algorithm}

\begin{theorem}
For any instance of \shm,
Algorithm \ref{alg:shm} gives a vector $q'$ and a matching $M$ such that $|q'(v)-q(v)|\le \ell -1$ for all $v\in V$, $\sum q(v)\le \sum q'(v)\le \sum q(v)+\ell -1$, such that $M$ is a stable matching with respect to capacities $q'$. Furthermore, after finding the initial dominating point $x^*$ in $\PP$ with Scarf's algorithm, the running time is polynomial.
\end{theorem}
\begin{proof}

We start by showing the bounds on $q'$. By the construction of the algorithm for $q'$, we get that $q'(v) = A_v z$ for the final vector $z\in \mathbb{Z}_+^{\EE}$ and also that $\sum_{v\in V}q'(v) = \sum_{e\in \EE}z[e]\cdot |e|$. By Claim~\ref{claim:hyp-saturated} it is also true for $q$ and $x^*$, that is $\sum_{v\in V}q(v) = \sum_{e\in \EE}x^*[e]\cdot |e|$.

As we only delete row $A_v$ if it satisfies $A_v(\lceil z\rceil - \lfloor z \rfloor )\le \ell$ and if $z$ has non-integer components then $A_v(\lceil z \rceil -  z) \ge 1$ and $A_v( z - \lfloor z \rfloor )\ge 1$ (using $A_vz=q(v)\in \mathbb{Z}_+$), it follows that $A_vz$ can only increase or decrease by at most $\ell -1$ during the algorithm (since every component of $z$ can only decrease until 0 or increase until 1), so the final $q'(v)$ capacities will only change by at most $\ell -1$ too.

By the fact that in each iteration we always choose $z$, such that it maximizes $\sum_{e\in \EE}z[e]\cdot|e|$, we have that $\sum_{v\in V} q(v)\le \sum_{v\in V} q'(v)$. Furthermore, it can only get strictly larger if we remove the aggregate capacity constraint. This happens only if there is exactly one $z[e]$ component that is fractional, let it be $z[f]$. This implies that  $\sum_{e\in \EE}z[e]\cdot|e|$ only change by strictly less than $|f|$ after this, and since it integral in both the start and the end, this implies that $\sum_{v\in V}q'(v)\le \sum_{v\in V}q(v)+\ell -1$.

The stability of the matching $M$ given by the output $z$ with respect to the $q'$ capacities defined as in Lemma \ref{lemma:hyp-1} is immediate, since the rounding procedure satisfies the requirements of Lemma \ref{lemma:hyp-1}.

Hence it only remains to show that the algorithm terminates, that is, if there was no row that could be deleted in the while loop, then $z$ already had to be integral.

Let us suppose the contrary. If all $A_v$ rows and the aggregate capacity row have been eliminated, then the remaining identity matrix is totally unimodular and the bounding vector is integral, so all extreme points are integer valued (see e.g., \cite{schrijver1998theory}), which would mean that $z$ is integral, contradiction. 
Hence, let us assume that there is such a remaining row. We have two cases:

\paragraph{Case a)} \textit{The aggregate capacity row has been eliminated. }

If the aggregate capacity row has been eliminated, then we know that there is at most $1$ fractional component of $z$ remaining. However, if there are no $A_v$ rows such that $A_v(\lceil z \rceil - \lfloor z \rfloor )\le \ell$, then $A_v(\lceil z \rceil - \lfloor z \rfloor )\ge \ell +1$ for any remaining row and from our assumption, we know that there is at least one. This means that, since each element of $A$ is 0 or 1, there has to be at least $\ell +1>1$ fractional components, as $\lceil z[e] \rceil - \lfloor z[e] \rfloor $ is only nonzero if $z[e]$ is fractional, contradiction. 

\paragraph{Case b):} \textit{The aggregate capacity row has not been eliminated.}

For simplicity, let us call $B$ the matrix consisting of $A$ and the aggregate capacity row at the bottom and $r$ be the vector which is $q$ appended by $\widetilde{q}(V)$.
Let $B =(B',B'')$ and $z= (z',z'')$, where $B'$ consists of the columns of $B$, where $z[e]$ is fractional and $B''$ consists columns of $B$ whose $z[e]$ components are integral. Similarly, $z'$ corresponds to the fractional components and $z''$ to the integral components of $z$, respectively.

 Since we can assume that $Bx=r$ in the beginning by Claim \ref{claim:hyp-saturated}, and we keep the rows not yet eliminated tight, we get that $B'z'=r - B''z''\in \mathbb{Z}^m$. We will use the following well known result (\cite{schrijver1998theory}):

\begin{lemma}
\label{lemma:extremepoint}
Let $y$ be an everywhere strictly positive extreme point of $\{ Qx\le  d,  x\ge 0\}$. Then the number of variables equals the maximum number of linearly independent tight rows of $Q$.
\end{lemma}
We use this lemma with \[
Q = \begin{pmatrix} B' \\ -B'\\ I \end{pmatrix}, \quad d = \begin{pmatrix} r-B''z'' \\ B''z''-r\\ 1 \end{pmatrix}
\] and $z'$. Since $z$ is an extreme point of $\{ Bx= q,0\le x\le 1 \}$, $z'$ is an everywhere positive extreme point of $\{ B'x=q-B''z'', 0\le x\le 1\} =\{ Qx\le d,x\ge 0 \}$. Indeed, if $z' = \lambda z_1' + (1-\lambda)z_2'$ for some $\lambda\in (0,1)$ and $z_1',z_2'\in \{ Qx\le d,x\ge 0\}$, then $z=\lambda (z_1',z'') + (1-\lambda ) (z_2',z'')$, with $(z_1',z''),(z_2',z'')\in \{ Bx=q,0\le x\le 1\}$, contradicting that $z$ is extreme. Observe that because $z'$ is everywhere fractional, there cannot be any tight rows at the bottom identity matrix.

Now we show that the number of linearly independent rows of $Q$ cannot be equal to the number of variables, which is equal to the number of components of $z'$. Note that the number of linearly independent tight rows of $Q$ is just the number of linearly independent rows of $B'$.

Give each $z'[e]$ component one credit. Then redistribute $z'[e]$'s credit to the rows of $Q$ the following way: if $e=\{ v_1,..,v_k\} $, then give $\frac{1}{k+1}$ credit to $B_{v_1}',..,B_{v_k}'$ and $\frac{1}{k+1}$ to $I_e$. If there are no $B_v'$ rows left, then $B'$ consists only of the aggregate capacity row, and we know there are at least $2$ fractional variables, so two components of $z'$, contradiction. Otherwise, there is a $B_v'$ row with $B_v'(\lceil z'\rceil -\lfloor z'\rfloor )\ge \ell +1$ (by $A_v(\lceil z\rceil -\lfloor z\rfloor )\ge \ell +1$), so there are at least $\ell +1$ fractional components of $z$, so $\ell +1$ components of $z'$. 

First, suppose there are only $\ell +1$. Then, by Lemma \ref{lemma:extremepoint}, there must be exactly $\ell +1$ independent rows in $B'$, and only one of them can be the aggregate capacity row, so there are at least $\ell \ge 2$ such rows of the form  $B_v'$. Also, because we cannot eliminate any of them, $B_v'(\lceil z\rceil -\lfloor z\rfloor )\ge \ell +1$ for each of them. This can only happen if there is a $1$ in every column in every $B_v'$ row meaning no two of them are not linearly independent, contradiction.

Hence, there are at least $\ell +2$ fractional components in $z$. This means, that the $I_e$ rows together got at least $\ell+2\cdot \frac{1}{\ell +1}>1$ credits. However, each $B_v'$ row got at least $\frac{1}{\ell +1}\cdot (\ell +1)=1$ credit too by $A_v(\lceil z\rceil -\lfloor z\rfloor )\ge \ell +1$ (and so $B_v'(\lceil z'\rceil -\lfloor z'\rfloor )\ge \ell +1$), so together with the aggregate capacity row they still got strictly more than 1 credit on average. Therefore, the number of components of $z'$ has to be strictly more than the number of rows of $B'$, which is at least the number of linearly independent tight rows of $Q$, contradiction. 

Hence, we obtained that if we cannot eliminate any row, then that necessarily means that $z$ is integral.

Finally, we show that the algorithm takes polynomial-time to terminate after the initial $x^*$ solution has been computed. This is because in each iteration, at least one of the rows is deleted, so there are at most $|V|$ iterations. Also, each iteration consists of checking some properties of $z$ and solving an LP, which can be done in polynomial time.
\end{proof}

Since in most practical applications of \shm, edges tend to be small, our result implies that stability can be ensured by only slight modifications in the capacities. 

In the case of the Stable Fixtures problem, which is the restriction of \shm\ to graphs, we can replace Scarf's algorithm by Tan's algorithm \cite{tan1991necessary} to obtain a fractional dominating solution $x^*$ in polynomial time. Hence, we obtain the following.

\begin{theorem}
There is a polynomial-time algorithm that, for a given instance of the Stable Fixtures problem, returns a modified capacity vector $q'$ and a matching $M$ such that $|q'(v)-q(v)|\le 1$ for all $v\in V$, $\sum\limits_{v\in V} q(v)\le \sum\limits_{v\in V} q'(v)\le \sum\limits_{v\in V} q(v)+1$ and $M$ is stable with respect to the capacities given by $q'$.
\end{theorem}

We remark that in the Stable Fixtures case, it is a well-known result that there exist capacities $(q'(v))_{v\in V}$ with $|q'(v)-q(v)|\le 1$ such that there is a stable matching with respect to the capacities given by $q'$ (see e.g., \cite{biro2024strong}), but our algorithm has the additional benefit of only changing the sum of the capacities by at most one too (there was no previous approach that could also guarantee this).

\section{Near-feasible Solutions for College Admission with Common Quotas}

In this section we proceed to the College Admission with Common Quotas (\cacq) problem.

Again, by analogous arguments as in the previous section, by tie-breaking we can suppose that the preferences are strict for our purposes. 

Let us suppose that each college is in at most $\ell$ sets with a common quota, (recall that for each college $c_i$, we assume that there exists a set $\{ c_i\}$ only containing $c_i$, so $\ell \ge 1$). 

Again, we start by showing how to find a stable fractional solution with Scarf's algorithm. Let the bipartite graph corresponding to the instance be $G=(\SS \cup \CC, \EE)$, where the edges $e\in \EE$ represent the acceptable student-college pairs $\{ s_j,c_i\}$.

We create a matrix $Q$, by adding a row for each college set $C_j, j\in N$ with a common quota $q(C_j)$ and also for each student $s_j\in \SS$, and a column for each edge $e\in \EE$. There is a 1 in row $C_j$, column $e= \{s_k,c_i\}$, if and only if $c_i\in C_j$, otherwise the corresponding entry is 0. There is a 1 in row $s_k$, column $e=\{s_j,c_i\}$ if and only if $s_k=s_j$, otherwise the entry is 0. 

The strict rankings of the rows are inherited from the corresponding master list/preference list with the addition that for a row $C_j$, if two columns have the same student $s_k$, then $C_j$ ranks these two columns according to the student $s_k$'s preferences. That is, for $c_i,c_j\in C_l$ and $s_k,s_p\in \SS$, row $C_l$ prefers $\{ s_k,c_i\}$ to $\{ s_p,c_j\}$ if and only if $s_k\succ_{C_l}s_p$ or $s_k=s_p$ and $c_i\succ_{s_k}c_j$.

The bounding vector $d$ is defined by setting it to $q(C_j)$ for a row $C_j$ and $1$ for a row $s_k$. The obtained polyhedron $\{ Qx\le d, x\ge 0\}$ with the created rankings satisfies the conditions of Scarf Lemma.

Note that the additional difficulty in this case is that the 'capacities' of the students are not allowed to be changed here, they must remain one.


\begin{lemma}
\label{lemma:cacq-1}
Let $x^*$ be the dominating solution found by running Scarf's algorithm on $\{ Qx\le d, x\ge 0\}$. Let $y\ge 0$ be an integer vector such that:
\begin{enumerate}
    \item If $x^*[e]=0$, then $y[e]=0$,
    \item If $Q_{s_k} x^*=1$, then $Q_{s_k} y=1$, otherwise $Q_{s_k} y\le 1$.
\end{enumerate}

Furthermore let $q'(C_j)=\begin{cases}
    Q_{C_j} y & \text{if } Q_{C_j} x^*=q(C_j) \\ \max\{ q(C_j), Q_{C_j} y\} &\text {otherwise}
\end{cases}$. Then, the matching $M$ given by its characteristic vector $y$ is stable and feasible with respect to $q'$.
\end{lemma}
\begin{proof}
First note that by $Q_{s_k}x^*\le 1$, we also have that $x^*\le 1$.
Hence, the feasibility of $M$ follows trivially from the definition of $q'$.  

Let us suppose that $M$ is not stable. Then, there exists a blocking edge $e=\{ s_k,c_i\}$. Since $x^*$ was a dominating solution, it must have dominated column $e$ in some row. 

If $x^*$ dominated it in row $s_k$, then $Q_{s_k}x^*=1$ and for each $f=\{s_k,c_{i'}\}$, with $x^*[f]>0$, it holds that $c_{i'}\succcurlyeq_{s_k}c_i$. By the two conditions on $y$, $s_k$ is assigned in $M$ and to a college $c_j$ with $x^*[\{s_k,c_{j}\}]>0$, so $e$ does not block, contradiction.

If $e$ was dominated in some row $C_j$ of $Q$, then $c_i\in C_j$, and this means that $Q_{C_j} x^*=q(C_j)$ and for each $s_{k'}$ and $c_{i'}\in C_j$ with $x^*[\{s_{k'},c_{i'}\}]>0$, we have that $s_{k'}\succcurlyeq_{C_j}s_k$. Hence, by the definition of $q'$ and the first condition on $y$, $C_j$ is saturated with weakly better students than $s_k$ (meaning $s_k$ can be in some other college in $C_j$ than $c_i$) and if $s_k$ is in $C_j$, then he is at a better college than $c_i$ by the definition of the row rankings, so $e$ does not block, contradiction.
\end{proof} 

We proceed to describe our iterative rounding algorithm, illustrated in Algorithm~\ref{alg:cacq}. 

The idea is similar. We first find a dominating solution $x^*$ using Scarf's algorithm and set $z=x^*$. Then, until the current $z$ is not integral, we find a row $C_j$ that we can safely remove without worrying about the quotas being violated too much in the future.

\begin{algorithm}[t]

	\SetAlgoNoLine
	\KwIn{An instance of \cacq.}
	\KwOut{A stable matching $M$ with respect to some common quotas $q'$ satisfying $|q(C_j)-q'(C_j)|\le 2\ell -1$ for all $C_j\in \{ C_1\dots,C_N\}$.}
	  Run Scarf's algorithm on $\PP =  \{ Qx\le  d, x\ge 0\}$ and obtain a dominating solution  $x^*$. 

    Initialize $z = x^*$.

	\While{$z$ is not integral}{
		\If{there is a non-tight row $C_j$, such that $Q_{C_j}(\lceil z \rceil - \lfloor z\rfloor )\le 2\ell -1$}{
			Delete it.
		}
		\ElseIf{there is a tight row $C_j$, such that $Q_{C_j}(\lceil z \rceil - \lfloor z\rfloor )\le 2\ell$}{
			Delete it. 
		}
		Consider the polyhedron given by the remaining constraints of $\PP$ along with the conditions $x[e] = z[e]$ if $z[e]$ is integer, and  $Q_{s_k} z=1$, if $Q_{s_k} x^*=1$. Then, find an extreme point $z^*$.

    	Update $z := z^*$.
    }
    Output the obtained matching $M$ corresponding to $z$ and the common quotas $q'$ as defined in Lemma \ref{lemma:cacq-1}.
	\caption{Iterative Rounding for \cacq}
	\label{alg:cacq}
\end{algorithm}

\begin{theorem} For any instance of \cacq, Algorithm \ref{alg:cacq} finds $q'(C_j)$ common quotas, such that $|q(C_j)-q'(C_j)|\le 2\ell -1$ for all set of colleges $C_j\in \{ C_1\dots,C_N\}$, and a stable and feasible matching $M$, with respect to $q'$. Furthermore, after finding the initial solution $x^*$ with Scarf's algorithm, the running time is polynomial.
\end{theorem}
\begin{proof}
By the definition of the common quotas $q'$, it is clear that $M$ gives a feasible matching.

We proceed to prove the bounds on $q'$.
Note that we only delete rows if they satisfy $Q_{C_j}(\lceil z\rceil - \lfloor z \rfloor )\le 2\ell -1$ or $Q_{C_j}(\lceil z\rceil - \lfloor z \rfloor )\le 2\ell$, if $Q_{C_j}z=q(C_j)$. Also, in the latter case, if $z$ has non-integer coordinates, where there is a 1 in row $Q_{C_j}$, then both $Q_{C_j}(\lceil z \rceil -  z) \ge 1$ and $Q_{C_j}( z - \lfloor z \rfloor )\ge 1$ as they are both positive integers, we get that $Q_{C_j}z$ can only increase or decrease by at most $2\ell -1$ during the algorithm. Therefore, the quotas $q'(C_j)$ only change by at most $2\ell -1$.


The stability of the matching $M$ given by the output $z$ with respect to the common quotas $q'$ defined as in Lemma \ref{lemma:cacq-1} is immediate, since the rounding procedure satisfies the requirements of Lemma \ref{lemma:cacq-1}.

Hence, it only remains to show that the algorithm terminates, that is, if there was no row that could be deleted in the while loop, then $z$ already had to be integral.

Let us suppose the contrary. If no row $Q_{C_j}$ remains, then the remaining matrix satisfies that there is only one nonzero element in each column $e=\{ s_k,c_i\}$ (in the row $Q_{s_k}$), that is a 1, so the remaining matrix is TU. Hence, the extreme point $z$ must already have been integral, contradiction.

Hence, assume that there are still some rows $Q_{C_j}$ in the remaining matrix.

Let $Q=(Q',Q'')$ and $(z',z'')$, where $Q'$ and $z'$ correspond to the fractional components of $z$, while $Q''$ and $z''$ correspond to the integral components of $z$.
Furthermore, let $d'=d-Q''z''$. It follows that $Q'z'\le d'$. 
Also, if no nonzero element remains in a row of $Q'$, then we delete it from $Q'$ and $d'$, as it is redundant.

We use Lemma \ref{lemma:extremepoint} with $Q'$ and $d'$ and the polyhedron $\{ Q'x\le d', x\ge 0\}$. Since $z$ is an extreme point of $\{ Qx\le d,x\ge 0 \}$, $z'$ is an everywhere positive extreme point of $\{ Q'x\le d',x\ge 0 \}$.

Now, we show that the number of linearly independent tight rows of $\{ Q'x\le d', x\ge 0\}$ cannot be equal to the number of variables, that is, the number of components of $z'$. 

Give each $z'[e]$ component one credit. Then redistribute $z'[e]$'s credit to the rows of $Q'$ the following way: if $e=\{ s_k,c_i\} $, and $c_i$ is included in $k\le \ell$ sets with a common quota, then give $\frac{1}{2}+\varepsilon$ credit to $Q_{s_k}$ and $\frac{1}{2k}-\frac{\varepsilon}{k}$ to each $Q_{C_j}$ with $c_i\in C_j$, where $\varepsilon$ is a sufficiently small number, such that $(2\ell +1)(\frac{1}{2\ell}-\frac{\varepsilon}{\ell})>1$. 

Then, each tight row $Q'_{s_k}z'=1$ must receive credits from at least two columns, hence it gets $2(\frac{1}{2}+\varepsilon )>1$ credits. As we cannot eliminate the row $Q_{C_j}$ for any remaining $C_j$, we get that each tight $Q_{C_j}$ row must have at least $2\ell +1$ corresponding fractional components by  $Q_{C_j}(\lceil z \rceil - \lfloor z\rfloor )> 2\ell$, so they get credits from at least $2\ell +1$ columns. By the choice of $\varepsilon$, this implies that they get strictly more than 1 credit too. Therefore, we get that each tight row must receive strictly more than one credit, so the number of linearly independent tight rows cannot be equal to the number of components of $z'$, contradiction.
Hence, if we cannot eliminate any row, then that necessarily means that $z$ is integral.

Finally, we show that the algorithm takes polynomial time to terminate after the initial $x^*$ solution has been computed. This is because in each iteration, at least one of the rows is deleted, so there is at most $\ell m$ iterations. Also, each iteration consists of checking some properties of $z$ and solving an LP, which can be done in polynomial time.
\end{proof}

We note that in most College Admission mechanisms, $\ell$ is quite small, i.e. at most $3$ or $4$. This is because such quotas usually only arise from restricting the total number of students that can be admitted to each department, each university and each specific major nationwide. Therefore, in such cases, our results imply that with only a slight modification in the quotas, a stable solution can be guaranteed. 

\section{Near-feasible Stable Multicommodity Flows}

Finally, we study the stable multicommodity flow framework. 
We give a simple algorithm that finds new capacities, that are close to the original ones, such that we can find an integer stable multicommodity flow with respect to the new capacities. Furthermore, this flow will be almost the same size as the original one, it can only change by less than one. 
Similarly to the other problems, as \citet{kiraly} proved, an integer stable multicommodity flow need not exist, even if every capacity is an integer.

Similarly, as in the previous cases, we claim that we can suppose that the preferences are strict. Indeed, if they are not strict, then by breaking the ties an arbitrary way, we can find a stable fractional solution and then round it to a near-feasible stable solution. Then, that same solution will be stable in the original instance too. This is again because if a blocking walk $W$ blocks a flow $f$ in the original instance, then it still must block $f$ after the tie-breaking too as strictly preferred relations remain strictly preferred.

First of all, as shown by \citet{csaji2022complexity}, the Stable Multicommodity Flow problem can be reduced to an instance of Scarf's lemma, and hence we can find a stable fractional multicommodity flow to begin with Scarf's algorithm.

We start with a simple claim.
\begin{claim}
\label{claim:smf-cyc}    
If a flow $f^j$ contains a fractional component $f^j(a)$, then there exists a cycle or an $s^jt^j$ path (in the undirected sense) such that all arcs in them have fractional $f^j$ values. 
\end{claim}
\begin{proof}
Consider the undirected graph $G$ corresponding to our directed network $D$, which contains an edge $\{ u,v\}$ for all arcs $uv\in \E$.

Let us consider the following greedy approach: we start from an edge that has fractional $f^j(e_1)$ (by abuse of notation $f^j(e)$ refers to $f^j(a)$ for the corresponding arc $a$ in $D$) value in $G$, and in case $s^j$ is adjacent to such an edge, we start from that. Then we always continue on a new edge $e_{i+1}$ of $G$ with fractional $f^j(e_{i+1})$ value, until we find a cycle or reach $t^j$. Suppose we cannot continue from $v_{i+1}$ after a fractional edge $e_i = \{ v_i,v_{i+1}\}$, but we have not yet found a cycle or an $s^jt^j$ path. Then, this means that every edge incident to $v_{i+1}\notin \{ s^j,t^j\}$ other than $e_i$ has integer value in $f^j$. Hence, in the flow $f^j$ of the network $D$, exactly one of $\sum_{a\in \delta (v_{i+1})}f^j(a)$ and $\sum_{a\in \rho (v_{i+1})}f^j(a)$ is integral, which contradicts the Kirchoff law of flows.
\end{proof}

Next, we prove a key Lemma for the rounding procedure.

\begin{lemma}
\label{lemma:smf-1}
Let $f=(f^1,\dots, f^k)$ be a fractional stable multicommodity flow. Let $g=(g^1,\dots, g^k)$ be an integer multicommodity flow, such that 
if $f^j(a)=0$, then $g^j(a)=0$.

Furthermore, let $$c'^j(a)=\left\{
	\begin{array}{ll}
		g^j(a)  & \mbox{if } f^j(a)=c^j(a) \\
		max\{ g^j(a),\lceil c^j(a)\rceil \} & \mbox{otherwise}
	\end{array}
\right.$$, $$c'(a)=\left\{
	\begin{array}{ll}
		\sum_{j\in [k]}g^j(a)  & \mbox{if } \sum_{j\in [k]} f^j(a)=c(a) \\
		max\{ \sum_{j\in [k]}g^j(a),\lceil c(a)\rceil \} & \mbox{otherwise}
	\end{array}
\right.$$  Then, $g$ is a stable multicommodity flow with respect to the capacities $c'$ and commodity specific capacities $c'^j$.

\end{lemma}

\begin{proof}

By the definition of the capacities, it is immediate that $g$ is feasible.
Assume for the contrary that $g$ is not stable and hence there exists a blocking walk $W=(v_1,a_1,v_2,\dots, e_{k-1},v_k)$  with respect to some commodity $j$.
Then, $g^j(a_i)\le c'^j(a_i)-1$ $\forall i$, which means that $f^j(a_i)<c^j(a_i)$ by the definition of $c'^j$.

If $v_1\ne s^j$, then there is an arc $a\in \delta (v_1)$ such that $g^j(a)\ge 1$ and $a_1\succ_{v_1}^ja$. By the condition on $g$, $f^j(a)>0$ too. 

If $v_k\ne t^j$, then there exists an arc $a'\in \rho (v_k)$, such that $g^j(a')\ge 1$ and $a_{k-1}\succ_{v_k}^ja'$, and again by the condition on $g$, $f^j(a')>0$ too.

Finally, if $\sum_{j\in [k]}g^j(a_i)=c'(a_i)$, then there exists some $j'\ne j$, such that $g^{j'}(a_i)\ge 1$ and $j\succ_{a_i}j'$. Hence, $f^{j'}(a_i)>0$ too.

Putting this together, we get that $W$ is a blocking walk for $f$ too, which contradicts the stability of $f$.
\end{proof}

We continue by describing the rounding algorithm for \smf. Given a directed graph $D=(V,\E)$, and a path or cycle $X$ in the undirected sense and a fixed orientation $\Vec{X}$ of $X$, we say that an arc $e\in X$ is a \emph{forward arc with respect to $\Vec{X}$}, if its direction is consistent with the orientation and a \emph{backward arc} otherwise. 

The algorithm starts by finding a fractional stable multicommodity flow $f$ and setting $g=f$.

Then, we go through each flow $g^j$ one-by-one and round them as follows. If we still have fractional values, then we find an $s^jt^j$ path or cycle $X$ of fractional edges. Then, we augment the flow with a sufficiently small $\varepsilon$ and iterate this until $g^j$ becomes integral. See Algorithm~\ref{alg:flow-round} for more details.

\begin{algorithm}[t]

	\SetAlgoNoLine
	\KwIn{An instance of \smf.}
	\KwOut{An integer stable multicommodity flow $f$ with respect to some capacities $c'$ satisfying $|c'(a)-c(a)|\le k -1$ for $a\in \E$.}
	  Find a stable fractional flow $f=(f^1,...,f^k)$ and set $g=f$.

   \For{ $j=1,...,k$} {
   \While{there is a fractional valued arc in $g^j$}{
    Find a cycle or an $s^jt^j$ path $X$ (in the undirected sense) as in Claim~\ref{claim:smf-cyc}.

    Fix an orientation $\Vec{X}$ of $X$. We choose this in a way such that if $X$ is an $s^jt^j$ path, then $\Vec{X}$ starts in $s^j$.

    Let $\varepsilon_1$ be the smallest number, such that increasing $g^j(a)$ on the forward arcs of $\Vec{X}$ and decreasing $g^j(a)$ on the backward arcs of $\Vec{X}$ with $\varepsilon_1$, at least one value becomes integral.

    Let $\varepsilon_2$ be the smallest number, such that increasing $g^j(a)$ on the backward arcs of $\Vec{X}$ and decreasing $g^j(a)$ on the forward arcs of $\Vec{X}$ with $\varepsilon_1$, at least one value becomes integral.
    

    \If{$X$ is a cycle}{
    \If{$\varepsilon_1<\varepsilon_2$}{
    Let $g^j(a) :=g^j(a)+\varepsilon_1$ for forward arcs and $g^j(a):=g^j(a)-\varepsilon_1$ for backward arcs with respect to $\Vec{X}$.
    }
   \Else{
  Let  $g^j(a) :=g^j(a)-\varepsilon_2$ for forward arcs and $g^j(a):=g^j(a)+\varepsilon_2$ for backward arcs with respect to $\Vec{X}$.
   }
    }
    \ElseIf{$|g^j|>|f^j|$}{
Decrease the flow with $\varepsilon_2$, that is, let $g^j(a) :=g^j(a)-\varepsilon_2$ for forward arcs and $g^j(a):=g^j(a)+\varepsilon_2$ for backward arcs with respect to $\Vec{X}$.
}
\ElseIf{$|g^j|\le |f^j|$}{
Increase the flow with $\varepsilon_1$, that is, let $g^j(a) :=g^j(a)+\varepsilon_1$ for forward arcs and $g^j(a):=g^j(a)-\varepsilon_1$ for backward arcs with respect to $\Vec{X}$.}

}
    }
    Output the obtained flow $g$, and capacities $c'$ as defined in Lemma \ref{lemma:smf-1}
	\caption{Iterative Rounding for \smf}
	\label{alg:flow-round}
\end{algorithm}

\begin{theorem}
For any instance of \smf, Algorithm~\ref{alg:flow-round}  finds capacities $c'$, $c'^j,j\in [k]$ with $\max|c(a)-c'(a)|\le k-1$ and $c^j(a)=c'^j(a)$ for all $a\in \E$, such that $g$ is a stable integral multicommodity flow with respect to the capacities $c'$ and commodity specific capacities $c'^j$.

Furthermore, each flow $g^j$ has almost the same size as $f^j$, meaning that $||f^j|-|g^j||<1$ for all $j$ and for the aggregate flow, $||f|-|g||<k-1$.

Similarly, with a slight modification to Algorithm~\ref{alg:flow-round}, we can also achieve that $||f^j|-|g^j||<2$ for all $j\in [k]$, and $||f|-|g||<1$.
\end{theorem}
\begin{proof}

First of all, as the algorithm sets the capacities according to Lemma~\ref{lemma:smf-1}, it is clear that $g$ is feasible with respect to the capacities. Since $g^j(a)$ remains 0, if $f^j(a)$ was 0, stability also follows by Lemma~\ref{lemma:smf-1}.

To see the bounds on the capacity changes, observe that for each $a\in \E$ and commodity $j\in [k]$, we never modify the value $g^j(a)$ if it was integral, and otherwise it can only change to one of the two closest integer values, by the choices of $\varepsilon_1$ and $\varepsilon_2$. This means also that $g^j(a)\le \lceil c^j(a)\rceil = c^j(a)$. Hence, if $c'^j(a)$ is would be updated to $g^j(a)\ne c^j(a)$ for an arc $a\in \E$ by the algorithm, then either $f^j(a)=c^j(a)$, so $c'^j(a)=c^j(a)$ or $g^j(a)>\lceil c^j(a)\rceil$, a contradiction in both cases.

Furthermore, because $c(a)$ is integer and $c'(a)\le g(a)= \sum_{j\in [k]} g^j(a)<f(a)+k$ is integer for any $a\in \E$, this means that the sum increases with $k-1$ at maximum. Also, if it decreases, then $c(a)=f(a)$, so $c'(a)=g(a)>f(a)-k$, so it decreases by at most $k-1$ too.



Observe that in the rounding procedure, the size $|g^j|=\sum_{a\in \delta (s^j)}g^j(a)$ of the flow $g^j$ only changes when $X$ is an $s^jt^j$ path. In every such step, it changes with less than one. Also, because we only decrease the size of the flow $g^j$, if its size is at least the size of $f^j$, and only increase if its size is less than the size of $f^j$, it cannot change with more than one, implying $||f^j|-|g^j||<1$ for any $j\in [k]$, as desired.
The bound on the change $||f|-|g||$ of the size of the aggregate flow $g$ immediately follows.

For the last statement, we modify the algorithm, such that when rounding the flow $g^j$ for an $s^jt^j$ path $X$, it keeps track of whether $|g|=\sum_{j\in [k]}\sum_{a\in \delta (s^j)}g^j(a)$ is greater or less than  $|f|=\sum_{j\in [k]}\sum_{a\in \delta (s^j)}f^j(a)$. Then, it increases the flow on it if $|g|<|f|$, unless $|g^j|> | f^j|+1$ and otherwise it decreases it, unless $|g^j|<|f^j|-1$.

Since in each rounding, the flow values change by strictly less than 1, this means that if we start to round a flow $g^j$ such that $|g|<|f|$ holds, then the size of $g^j$ cannot decrease by the end of iteration $j$, so at the end of that iteration, $|g|$ cannot get smaller. Similarly, if $|g|\ge |f|$ by the start of the iteration, then $|g^j|$ and thus $|g|$ cannot get larger. Hence, $||f|-|g||<1$ as desired. Furthermore, it is also easy to see that $|f^j-g^j|<2$ for all $j\in [k]$.

The algorithm starts with finding a fractional stable flow, which of course may not have polynomial running time. The rest of the algorithm is easily seen to be polynomial, as the number of iterations is bounded by $k\cdot |\E|$ and each iteration can be done in polynomial time.
\end{proof}\section{Conclusions}

In this paper, we have demonstrated that near-feasible stable solutions can always be found for a variety of NP-hard stable matching problems, including the Stable Hypergraph Matching problem, the College Admission problem with Common Quotas, and the Stable Multicommodity Flow problem. Our results establish that by allowing minimal modifications to capacity constraints, we can ensure stability in these settings without violating the fundamental properties of the underlying models.

We used an iterative rounding algorithm that systematically adjusts capacities while preserving stability. Our approach leverages Scarf’s algorithm to compute an initial fractional stable solution and iteratively rounds it to obtain an integral stable solution with bounded capacity adjustments. 
Moreover, in the case of Stable Fixtures, our method provides a polynomial-time guarantee for achieving a stable matching with minimal capacity modifications.

Our findings have significant implications for real-world applications in market design, resource allocation, and network flow optimization. In particular, they offer constructive solutions for policymakers and system designers, ensuring that slight adjustments to institutional constraints can lead to stable and implementable outcomes. Future research directions include refining our bounds for specific problem instances, exploring alternative rounding techniques, and extending our methods to additional constrained matching and flow problems.

Our work highlights the practical viability of near-feasible stable solutions, reinforcing the idea that stability can be preserved even in highly constrained and complex settings with only minimal systemic changes.

\bibliographystyle{ACM-Reference-Format}
\bibliography{main}

\end{document}